\documentclass[12pt,a4paper,reqno]{amsart}
\usepackage{lmodern,color}
\usepackage[english]{babel}
\usepackage[T1]{fontenc}
\usepackage[margin=1in]{geometry}
\usepackage[title]{appendix}
\usepackage{amsthm,amssymb,amsmath,mathtools,diffcoeff,enumerate,mathrsfs,hyperref,array,ragged2e}
\newtheorem{theorem}{Theorem}
\newtheorem{corollary}{Corollary}
\newtheorem{remark}{Remark}
\newtheorem{lemma}{Lemma}
\numberwithin{equation}{section}
\numberwithin{theorem}{section}
\numberwithin{corollary}{section}
\numberwithin{remark}{section}
\numberwithin{lemma}{section}

\newcolumntype{P}[1]{>{\centering\arraybackslash}p{#1}}
\usepackage{caption}
\begin{document}
\pagenumbering{arabic}

\author{Garima Agrawal}
\address{IISER Pune, India}
\email{agrawal.garima@students.iiserpune.ac.in}
\author{Anindya Goswami}
\address{IISER Pune, India}
\email{anindya@iiserpune.ac.in}
\title{A semi-Markovian approach to Model the tick-by-tick dynamics of stock price}
\thanks{The research by the first author is supported by IISER Pune research fellowship. The research by the second author was supported in part by the SERB MATRICS (MTR/2017/000543), DST FIST (SR/FST/MSI-105), NBHM 02011/1/2019/NBHM(RP)R\&D-II/585, and DST/INT/DAAD/P-12/2020. }
\maketitle

\begin{abstract}
We model the stock price dynamics through a semi-Markov process obtained using a Poisson random measure. We establish the existence and uniqueness of the classical solution of a non-homogeneous terminal value problem and we show that the expected value of stock price at horizon can be obtained as a classical solution of a linear partial differential equation that is a special case of the terminal value problem studied in this paper. We further analyze the market making problem using the point of view of an agent who posts the limit orders at the best price available. We use the dynamic programming principle to obtain a HJB equation. In no-risk aversion case, we obtain the value function as a classical solution of a linear pde and derive the expressions for optimal controls by solving the HJB equation. 
\end{abstract}

\section{Introduction}
Market microstructure is a branch of finance that deals with the details of how the trading occurs in the financial market. According to the National Bureau of Economic Research(NBER), the study of market microstructure is concerned with the theoretical, empirical, and experimental research on the economics of security markets, including the role of information in the price discovery process, the definition, measurement, control, and determinants of liquidity and transactions costs, and their implications for the efficiency, welfare, and regulation of alternative trading mechanisms and market structures\cite{nber}.\par
In this work, we aim to study the market microstructure by modelling the tick by tick dynamics of stock price via a semi-Markov process where the duration between two ticks may not be exponential\cite{pham1}.

\section{Price Dynamics}
\subsection{Overview of the Model} The ask or bid prices are quotations by agents who want to sell or buy certain units of a financial asset respectively. The lowest ask price and the highest bid price are called the \emph{`best ask'} and \emph{`best bid'} prices respectively. At a given time the \emph{bid-ask spread} corresponds to an interval, whose upper limit is the best ask price and the lower limit is the best bid price. There are two types of orders prominent in the financial market which we consider in our model known as the \emph{limit order} and the \emph{market order}. The \emph{limit order} is posted by a trader who wishes to buy or sell at a price different than the current best price and it is executed only when a potential sell or buy order arrives matching the limit order.\\\par

On the contrary, the \emph{market order} is placed by a potential buyer or seller who wishes to buy or sell at the current best ask or best bid price and the order is exercised immediately. On the basis of `size', the market orders can be termed \emph{small} or \emph{big}. The small orders do not cause any change in the price of the asset. On the other hand, the large market orders may exhaust all the liquidity available in the market. At these occasions the stock price moves down or up if the market order is of sell or buy type respectively. For the purpose of modelling, from now onward `price' refers to the \emph{mid-price} of the stock which is the arithmetic mean of the best bid and the best ask price. \\\par

We model the return dynamics of price through a semi-Markov process which changes its state whenever the stock price goes up or down. In the most natural way, the state space can be thought of having just two elements each one representing positive and negative price jump direction. Since the direction of the price jump can be same in two consecutive jumps, we consider an auxiliary state process where the state space is not binary rather it has four states but a given state communicates to only two other states that are specified beforehand for that state. This helps us to obtain the \emph{state} process as a semi-Markov process with an embedded Markov chain that covers all the possibilities of price transition exactly once.\\ \par 

The \emph{price} process is a pure jump process and we assume that the absolute value of simple rate of return always remains constant denoted by $\delta$. We also assume that the bid-ask spread is equal to $2p\delta$ where $p$ is the current mid-price. The price process together with state process and its associated \emph{age} process form an augmented Markov process obtained as a solution of stochastic integral equations involving integration with respect to a Poisson random measure. Next we formulate the required mathematical set up and express the stock price dynamics formally in mathematical terms.

\subsection{Construction of underlying process}
Assume that $X$ is a Markov chain on the binary set $\{o,e\}$ with transition matrix $P$ and the auxiliary chain $X'$ is on $\{1,2,3,4\}$, with transition matrix $\begin{pmatrix} O_{2\times 2}  & P \\  P  &  O_{2\times 2} \end{pmatrix}$. If the projection map $\varphi\colon \{1,2,3,4\} \to \{o,e\}$ is such that $\varphi(1)= \varphi(3)= o $, and $\varphi(2)= \varphi(4) =e$, then $X$ and $\varphi (X')$, the projection of auxiliary chain, have identical law. Here, $X'$ does not transit to the present state even if $X$ does. This is further clarified by the diagram in Figure \ref{fig:TPM}. 
\begin{figure}[h]
    \centering
    \includegraphics[width=0.6\linewidth]{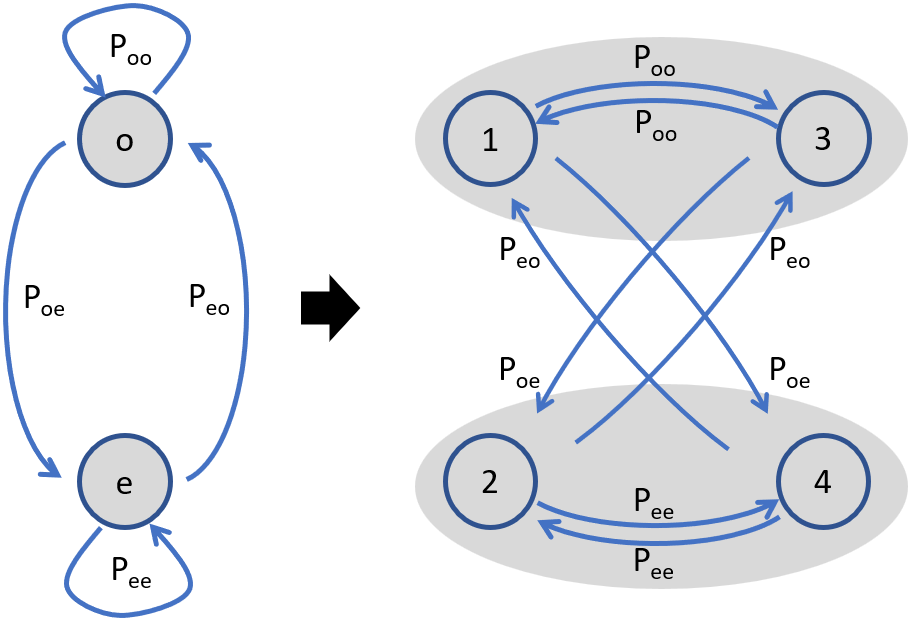}
    \caption{The diagrammatic representation of transition probability matrices of a binary Markov chain and its auxiliary process}
    \label{fig:TPM}
\end{figure}
We recall that transition to the same state by a continuous-time process on a discrete state-space is imperceptible. However, in the above manner, such a process can be formulated as the projection of its auxiliary process. We make use of the auxiliary chain for modelling the price direction process.\\ \par

Let $(\Omega,\mathscr{F},\{{\mathscr{F}}_t\}_{t\geq 0},\mathbb{P})$ be a complete filtered probability space satisfying the usual hypothesis. The state space $\{1,2,3,4\}$ is denoted by $\mathcal{X}$. We define an equivalence relation `$\equiv$' on $\mathcal{X}$ such that for $i,j\in \mathcal{X}$, we say $i\equiv j$ if and only if both belong to either $\{1,2\}$ or $\{3,4\}$ and we say $i \not\equiv j$ otherwise. We denote the set $\{(i,j)\in \mathcal{X}^2 \mid i\not\equiv j\}$ by ${\mathcal{X}}^{'}$. For each $\nu \in \{+,-\}$, let $h_{\nu}$ be a non-negative measurable function defined on $[0,\infty)$ and for each $s \in [0,\infty)$, let $H_+(s)$ and $H_-(s)$ be two right open, left closed non-overlapping intervals with length equal to $h_+(s)$ and $h_-(s)$ respectively. We denote $H_+(s) \cup H_-(s)$ and $h_+(s)+h_-(s)$ by $H(s)$ and $h(s)$ respectively for each $s \geq 0$.~\\ \par
Set maps $\alpha:\mathbb{Z}\to \{-1, 1\}$ and $sign : \{-1, 1\} \rightarrow \{-, +\}$  such that $\alpha(i) = {(-1)}^i$, $sign(-1) = -$ and $sign(+1) = +$. Let $\bar{\alpha}(i)$ denote $sign(\alpha(i))$. For the brevity of notation, for each $s \geq 0,(i,j) \in {\mathcal{X}}^{'}$ we denote $H_{\bar\alpha(i+j)}(s)$ and $h_{\bar\alpha(i+j)}(s)$ by $H_{ij}(s)$ and $h_{ij}(s)$ respectively. Now we define the real valued measurable maps $g_1$, $g_2$ and $g_3$ on the set $\mathcal{X}\times[0,\infty)\times\mathbb{R}$ as follows:
\begin{align}
&g_1(i,s,z) =s1_{H(s)}(z)\label{agefn}\\
&g_2(i,s,z) = \sum_{j \not\equiv i} (j-i)1_{H_{ij}(s)}(z)\label{statefn}\\
&g_3(i,s,z) = \sum_{j \not\equiv i} \delta \alpha(j)1_{H_{ij}(s)}(z) \label{pricefn}
\end{align}
where $1_A$ denotes the indicator function of set $A$ and $\delta>0$.\\~\\
Let $S_0$, $I_0$ and $P_0$ be ${\mathscr{F}}_0$-measurable independent random variables. We make the following assumptions which are applicable throughout this paper.
\begin{itemize}
\item[(A1)] $h_+$ and $h_-$ are continuously differentiable maps.
\item [(A2)]$\sup\limits_{\nu \in \{+,-\},y\in [0,\infty)} h_{\nu}(y)<\infty$ and we denote it by $c_1$.
\item [(A3)]For each $\nu \in \{+,-\}, \lim_{y\rightarrow\infty}\int_{0}^{y}h_{\nu}(v)dv = \infty.$
\item[(A4)] $h(y)>0$ for every $y\in[0,\infty)$.
\item[(A5)] $E[P_0]<\infty.$
\item[(A6)] The probability distributions of the random variables $S_0$, $I_0$ and $P_0$ have full support.
\end{itemize}~\\
Let $\wp(ds,dz)$ be a Poisson random measure defined on $[0,\infty) \times \mathbb{R}$ with intensity measure equal to the product Lebesgue measure on $[0,\infty) \times \mathbb{R}$ and it is adapted to the filtration $\{{\mathscr{F}}_t\}$. Let us consider the following processes:
\begin{align}
&S_t = S_0+t-\int_{(0,t]\times\mathbb{R}}g_1(I_{v^-},S_{v^-},z)\wp(dv,dz)\label{age}\\
&I_t =I_0 + \int_{(0,t]\times\mathbb{R}}g_2(I_{v^-},S_{v^-},z)\wp(dv,dz)\label{state}\\
&P_t = P_0+\int_{(0,t]\times\mathbb{R}}P_{v^-}\; g_3(I_{v^-},S_{v^-},z)\wp(dv,dz)\label{price}
\end{align}~\\
It is known that under the assumption (A2), there exists a pathwise unique strong solution $(P,I,S) :=\{(P_t, I_t, S_t)\}_{t\ge 0}$ to the above system of equations which is right continuous having left limits, locally bounded and strongly Markov (see Theorem 3.4 (p-474) of \cite{cinlar2011probability}). It is evident from the above equations that $P$, $I$ are piecewise constant and $S$ is a piecewise linear process and these processes have jump discontinuities at the same time.\par
\subsection{Law of the underlying process}
It can be seen that under the assumptions (A2) and (A3), $I$ is a semi-Markov process with $S$ being the `age' process \cite{mkghoshsaha}. We now define few parameters associated with the process $(I,S)$.
Further for each $(i,j) \in {\mathcal{X}}^{'}$, we define a map $p_{ij}$ on the set $[0,\infty)$ such that for $y \geq 0$,
\begin{equation}
p_{ij}(y) = \frac{h_{ij}(y)}{h(y)}\label{p_{ij}(y)}. 
\end{equation}
Due to assumption (A4), the map $p_{ij}$ is well-defined. Now let $F$ and $f$ be two maps defined on $[0,\infty)$ such that for $y \geq 0,$
\begin{align}
&F(y) = 1-e^{-\int_0^y h(v)dv},\label{F(y)}\\
&f(y) =  h(y)e^{-\int_0^y h(v)dv}.\label{f(y)}
\end{align}~\\
The assumption (A1) affirms that $F$, its derivative on the set $(0,\infty)$ given by $f$, and $p_{ij}$ for all $(i,j) \in {\mathcal{X}}^{'}$ are continuously differentiable functions. Owing to the assumptions (A3), (A4) and the expression of $F$, it follows that $F$ is a strictly increasing function, it never attains the value 1 for any $y \in [0,\infty)$ and $\lim_{y\rightarrow\infty} F(y)=1$.~\\ \par
It is known that $p_{ij}(y)$ is the conditional probability that the process transits next to the state `$j$' given that the process is currently in state `$i$' and the next transition happens after `$y$' unit of time has been elapsed since the previous transition. Further, $F$ denotes the conditional cumulative distribution function of the holding time associated with the process $I$ given the current state \cite{mkghoshsaha}. It turns out that $F$ is independent of the state in this scenario.~\\
We also see that for $y \geq 0$,
\begin{equation}
\frac{f(y)p_{ij}(y)}{1-F(y)}=h_{ij}(y).\label{hij}
\end{equation}
Now let $\psi$ be a real-valued continuous map defined on the set $[0,\infty)\times\mathcal{X}\times[0,\infty)$ which is compactly supported as well as continuously differentiable function in the third variable. Let $G^\psi$ be real valued map on $(0,T]\times \mathbb{R}$ such that for $(v,z)\in (0,T]\times\mathbb{R}$,
\begin{align}
G^\psi(v,z):=&\psi(P_{{v}^{-}}(1+g_3(I_{{v}^{-}},S_{{v}^{-}},z)),I_{{v}^{-}}+g_2(I_{{v}^{-}},S_{{v}^{-}},z),S_{{v}^{-}}-g_1(S_{{v}^{-}},I_{{v}^{-}},z))\nonumber\\
&-\psi(P_{{v}^{-}},I_{{v}^{-}},S_{{v}^{-}}). \label{gpsi} 
\end{align}\\
We now calculate the infinitesimal generator of the process $(P,I,S)$. 
Using It\^{o}'s formula,
\begin{align}
&\psi(P_t,I_t,S_t) - \psi(P_0,I_0,S_0) = \int_{0}^{t}\diffp{\psi}{s}(P_{v^-},I_{v^-},S_{v^-})dv+\int_{0}^{t} \int_{\mathbb{R}}G^\psi(v,z)\wp(dz,dv).\label{genpis1}
\end{align}
Again,
\begin{align}
&\int_{0}^{t} \int_{\mathbb{R}}G^\psi(v,z)dz\;dv \nonumber\\
&=\int_{0}^{t} \int_{\mathbb{R}}\bigg(\psi\bigg(P_{{v}^{-}}(1 + \sum_{j \not\equiv I_{{v}^{-}}} \delta \alpha(j)1_{H_{I_{v^{-}}j}(S_{{v}^{-}})}(z)), I_{{v}^{-}} + \sum_{j \not\equiv I_{{v}^{-}}} (j-I_{{v}^{-}})1_{H_{I_{v^{-}}j}(S_{{v}^{-}})}(z)\nonumber\\ 
& \qquad\qquad\qquad\qquad\qquad\qquad\qquad \qquad\qquad S_{{v}^{-}}(1- 1_{H(S_{{v}^{-}})}(z))\bigg)-\psi(P_{{v}^{-}},I_{{v}^{-}},S_{{v}^{-}})\bigg)dz\;dv\nonumber\\
&=\int_{0}^{t} \sum_{j\not\equiv I_{{v}^{-}}}h_{I_{v^{-}}j}(S_{{v}^{-}})(\psi(P_{{v}^{-}}(1+\delta\alpha(j)),j,0)-\psi(P_{{v}^{-}},I_{{v}^{-}},S_{{v}^{-}}))dv\label{genpis2}
\end{align}
is finite using the boundedness of $\psi$ and assumption (A2). Hence we can rewrite the last term of \eqref{genpis1} as 
\begin{equation}
\int_{0}^{t} \int_{\mathbb{R}}G^\psi(v,z)\wp(dz,dv) =  \int_{0}^{t} \int_{\mathbb{R}}G^\psi(v,z)\tilde{\wp}(dz,dv) + \int_{0}^{t} \int_{\mathbb{R}}G^\psi(v,z)dz\;dv \label{genpis3}    
\end{equation}
where $\tilde{\wp}$ is the compensated Poisson random measure. Using \eqref{genpis1} - \eqref{genpis3}, we infer that the infinitesimal generator $\mathscr{A}$ of the process $(P,I,S)$ is given by
\begin{equation}
\mathscr{A}\psi(p,i,s) = \diffp{\psi}{s}(p,i,s)+\sum_{j\not\equiv i}h_{ij}(s)(\psi(p(1+\delta\alpha(j)),j,0)-\psi(p,i,s)).\label{genpis}
\end{equation}

\section{An applicatory non-homogeneous terminal value problem}
\mbox\\~\\
For a Borel subset $B$ of $\mathbb{R}^d$, let $C(B)$ denote the class of all real-valued continuous maps on $B$. 
Let $\mathcal{D} := [0,T] \times [0,\infty) \times \mathcal{X} \times [0,\infty)$ and ${\mathcal{D}}^{\mathrm{o}} := (0,T) \times (0,\infty) \times \mathcal{X} \times (0,\infty)$. Let $\mathcal{D}_{t,s}$ denote the operator given by $\mathcal{D}_{t,s} = \left(\diffp{}{t}+\diffp{}{s}\right)$, and $cl(\mathcal{D}_{t,s})$ denote the class of functions in $C\mathcal{(D^{\mathrm{o}})}$ which are in the domain of the map $\mathcal{D}_{t,s}$. Consider the following set $V$ where
\[\quad V := \left\{\phi \in C\mathcal{(D)} \bigg\vert \sup\limits_{(t,p,i,s)\in\mathcal{D}}\frac{|\phi(t,p,i,s)|}{1+p}<\infty\right\}.\]
$(V,{||.||}_V)$ is a Banach space where for $\phi \in V$, ${||\phi||}_V = \sup\limits_{(t,p,i,s)\in\mathcal{D}}\frac{|\phi(t,p,i,s)|}{1+p}$. 
 
\begin{theorem}\label{thm1}
Consider the map $\mathcal{T}$ on $V$, given by
\begin{align}\label{thm1eq1}
\mathcal{T}\phi(t,p,i,s)&:=g(p)\frac{1-F(T-t+s)}{1-F(s)}\nonumber\\
&+\int_{t}^{T}\frac{f(v-t+s)}{1-F(s)}\sum_{j \not\equiv i}{p_{ij}(v-t+s)\phi(v,p(1+\delta\alpha(j)),j,0)}dv\nonumber\\
&+\int_{t}^{T}\frac{1-F(v-t+s)}{1-F(s)}w(v,p,i,v-t+s)dv
\end{align}
for each $\phi\in V$, where $w\in V$ and $g$ is a continuous function on $[0,\infty)$ with at most linear growth. 
Then the following hold.
\begin{enumerate}
\item $\mathcal{T}$ maps $V$ to $V$.
\item $\mathcal{T}$ is a contraction.
\item The following equation has a unique solution in $V$
\begin{align} \label{thm1eq2}
\mathcal{T}\phi = \phi.
\end{align} 
\item $\mathcal{T}\phi\mid_{\mathcal{D^{\mathrm{o}}}} \in cl(\mathcal{D}_{t,s})$.
\item If $\diffp{w}{s}$ exists and belongs to $C({\mathcal{D}}^{\mathrm{o}})$ then $\mathcal{T}\phi\mid_{\mathcal{D^{\mathrm{o}}}} \in C^{1,1}\mathcal{(D^{\mathrm{o}})}$.
\end{enumerate}
\end{theorem}
\begin{proof} 
For fixed $\phi \in V$ and $(t,p,i,s)\in \mathcal{D}$ we first rewrite 
\begin{equation}
\mathcal{T}\phi(t,p,i,s)=\bigg({\mathcal{T}}_{1}+\sum_{j \not\equiv i}{\mathcal{T}}_{2}^j+{\mathcal{T}}_{3}\bigg)(t,p,i,s) \label{T1T2T3}    
\end{equation}\
where
\begin{align}
&{\mathcal{T}}_{1}(t,p,i,s)= g(p)\frac{1-F(T-t+s)}{1-F(s)},\label{T1} \\
&{\mathcal{T}}_{2}^j(t,p,i,s)=\int_{t}^{T}\frac{(fp_{ij})(v-t+s)}{1-F(s)}\phi(v,p(1+\delta\alpha(j)),j,0)dv,\label{T2}\\
&{\mathcal{T}}_{3}(t,p,i,s)=\int_{t}^{T}\frac{1-F(v-t+s)}{1-F(s)}w(v,p,i,v-t+s)dv.\label{T3}
\end{align}~\\
\textbf{Continuity of $\mathcal{T}\phi$}\\~\\
We note that since $F$ and $g$ are continuous functions, ${\mathcal{T}}_{1}$ is continuous.\\
Next we prove ${\mathcal{T}}_{2}^{j}$ is continuous. First we check the continuity of ${\mathcal{T}}_{2}^{j}$ with respect to $t$-variable. Now for $(t,p,i,s)\in \mathcal{D}$ and $\varepsilon \in (-t,T-t),\; \varepsilon \neq 0$, consider
\begin{align} 
&{\mathcal{T}}_{2}^{j}(t+\varepsilon,p,i,s)-{\mathcal{T}}_{2}^{j}(t,p,i,s)\nonumber\\
&=\int_{t+\varepsilon}^{T}\frac{(fp_{ij})(v-t-\varepsilon+s)}{1-F(s)}\phi(v,p(1+\delta\alpha(j)),j,0)dv\nonumber \\
&-\int_{t}^{T}\frac{(fp_{ij})(v-t+s)}{1-F(s)}\phi(v,p(1+\delta\alpha(j)),j,0)dv\nonumber\\
&=\frac{1}{1-F(s)}\bigg(\int_{t+\varepsilon}^{T}((fp_{ij})(v-t-\varepsilon+s)-(fp_{ij})(v-t+s))\phi(v,p(1+\delta\alpha(j)),j,0)dv \nonumber\\ 
&\quad -\int_{t}^{t+\varepsilon}(fp_{ij})(v-t+s)\phi(v,p(1+\delta\alpha(j)),j,0)dv\bigg). \label{c2t}
\end{align}
Using the fact that $f,p_{ij}$ and $\phi(\cdot, p(1+\delta\alpha(j)), j,0)$ are continuous and the domain of integration is contained in a compact set, therefore using bounded convergence theorem we get
\begin{align*}
&\lim_{\varepsilon\rightarrow 0} {\mathcal{T}}_{2}^{j}(t+\varepsilon,p,i,s)-{\mathcal{T}}_{2}^{j}(t,p,i,s)=0.  
\end{align*}
Hence the continuity of ${\mathcal{T}}_{2}^{j}$ with respect to $t$-variable has been proved.\\~\\
Next we prove the continuity of ${\mathcal{T}}_{2}^{j}$ with respect to $s$-variable. For $(t,p,i,s)\in \mathcal{D}$ and $\varepsilon \in (-s,s), \varepsilon \neq 0$, consider
\begin{align}
&{\mathcal{T}}_{2}^{j}(t,p,i,s+\varepsilon)-{\mathcal{T}}_{2}^{j}(t,p,i,s)=\nonumber\\
&\int_{t}^{T}\bigg(\frac{(fp_{ij})(v-t+s+\varepsilon)}{1-F(s+\varepsilon)}-\frac{(fp_{ij})(v-t+s)}{1-F(s)}\bigg)\phi(v,p(1+\delta\alpha(j)),j,0))dv.\label{c2s}  
\end{align}
Since $f$, $p_{ij}$, $F$ and $\phi(\cdot, p(1+\delta\alpha(j)), j,0)$ are continuous, proceeding similarly as above, we see that
\begin{equation}
\lim_{\varepsilon\rightarrow 0}{\mathcal{T}}_{2}^{j}(t,p,i,s+\varepsilon)-{\mathcal{T}}_{2}^{j}(t,p,i,s)= 0. \nonumber
\end{equation}
We have proved that ${\mathcal{T}}_{2}^{j}$ is continuous with respect to $s$-variable.
We now prove the continuity of ${\mathcal{T}}_{2}^{j}$ with respect to $p$-variable. Again for $(t,p,i,s)\in \mathcal{D}$ and $\varepsilon \in (-p,p),\; \varepsilon\neq0$, we have
\begin{align*}
&{\mathcal{T}}_{2}^{j}(t,p+\varepsilon,i,s)-{\mathcal{T}}_{2}^{j}(t,p,i,s)\\
&=\int_{t}^{T}\frac{(fp_{ij})(v-t+s)}{1-F(s)}(\phi(v,(p+\varepsilon)(1+\delta\alpha(j)),j,0)-\phi(v,p(1+\delta\alpha(j)),j,0))dv.
\end{align*}
Since $\phi$ is continuous, there is a neighborhood $U_p$ of $p$ such that $\phi$ is bounded on the set $[t,T]\times (1+\delta\alpha(j))U_p \times \{j\}\times \{0\}$. Again as $f$ and $p_{ij}$ are also continuous and the integration is over a compact set, using bounded convergence theorem we have
\begin{equation}
\lim_{\varepsilon\rightarrow 0}{\mathcal{T}}_{2}^{j}(t,p+\varepsilon,i,s)-{\mathcal{T}}_{2}^{j}(t,p,i,s)= 0.\nonumber
\end{equation}
Thus, we have proved the continuity of ${\mathcal{T}}_{2}^{j}$ with respect to $p$-variable.\\~\\
Summarising all the three cases we conclude ${\mathcal{T}}_{2}^{j} \in C\mathcal{(D)}$. Since $F$ and $f$ are continuous maps, continuity of ${\mathcal{T}}_{3}$ can be proved in the similar way as was the case in proving continuity of ${\mathcal{T}}_{2}^{j}$. Since the continuity of $\mathcal{T}_{1},{\mathcal{T}}_{2}^{j}$ and $\mathcal{T}_{3}$ has been established, therefore using \eqref{T1T2T3} we conclude that  $\mathcal{T}\phi$ is continuous.\\~\\
Now, in order to prove the assertion (1),
We note that for the case $\phi=0$, for all $(t,p,i,s) \in \mathcal{D}$
\begin{align*}
&\frac{\vert\mathcal{T}0(t,p,i,s)\vert}{1+p}\\
&=\bigg\vert\frac{g(p)}{1+p}\left(\frac{1-F(T-t+s)}{1-F(s)}\right) +\frac{1}{1+p}\left(\int_{t}^{T}\frac{1-F(v-t+s)}{1-F(s)}f(v,p,i,v-t+s)dv\right)\bigg\vert\\
&\leq \frac{\vert g(p)\vert}{1+p}+T{\Vert f\Vert}_{V}
\end{align*}
using the fact that $F$ is non-decreasing.\\ 
Since $g$ has at most linear growth, we have
$\sup\limits_{(t,p,i,s)\in\mathcal{D}}\frac{\vert\mathcal{T}0(t,p,i,s)\vert}{1+p}<\infty$. Hence, $\mathcal{T}0 \in V$.\\ 
Now for any $\phi_1, \phi_2 \in V$ and $(t,p,i,s) \in \mathcal{D}$,
\begin{align*}
&\frac{\vert\mathcal{T}\phi_1(t,p,i,s)-\mathcal{T}\phi_2(t,p,i,s)\vert}{1+p}=\\
&\bigg\vert\int_{0}^{T-t}\sum_{j \not\equiv i}{p_{ij}(v+s)\frac{(\phi_1(t+v,p(1+\delta\alpha(j)),j,0)-\phi_2(t+v,p(1+\delta\alpha(j)),j,0))}{1+p}\frac{f(v+s)}{1-F(s)}dv}\bigg\vert\\
&\leq \int_{0}^{T-t}\sum_{j \not\equiv i}{p_{ij}(v+s)\frac{|\phi_1(t+v,p(1+\delta\alpha(j)),j,0)-\phi_2(t+v,p(1+\delta\alpha(j)),j,0)|}{1+p}\frac{f(v+s)}{1-F(s)}dv}\\
&\leq{\Vert\phi_1-\phi_2\Vert}_{V}\int_{0}^{T-t}\frac{f(v+s)}{1-F(s)}dv\\
&={\Vert\phi_1-\phi_2\Vert}_{V}\left(\frac{F(T-t+s)-F(s)}{1-F(s)}\right)\\
&<{\Vert\phi_1-\phi_2\Vert}_{V}
\end{align*}
since $F$ is non-decreasing and $F(\cdot)<1$. Hence, 
\begin{equation}
\sup\limits_{(t,p,i,s)\in\mathcal{D}}\frac{\vert\mathcal{T}\phi_1(t,p,i,s)-\mathcal{T}\phi_2(t,p,i,s)\vert}{1+p} < {\Vert\phi_1-\phi_2\Vert}_{V}. \label{contrac}   
\end{equation}
For each $\phi\in V$, the above inequality and continuity of $\mathcal{T}\phi$ imply that $\mathcal{T}\phi-\mathcal{T}0$ is in $V$. Again, as $V$ is a linear space and $\mathcal{T}0$ is in $V$, $\mathcal{T}\phi$ is also in $V$ for each $\phi \in V$.
Therefore, we have shown $\mathcal{T}$ maps $V$ to $V$. Then the equation \eqref{contrac} implies that $\mathcal{T}$ is a contraction. Hence the assertions (1) are (2) are proved.  As $V$ is a Banach space and a solution of \eqref{thm1eq2} is a fixed point of the contraction map $\mathcal{T}$ on $V$, using Banach Fixed Point theorem we conclude that the equation \eqref{thm1eq2} has a unique solution in $V$. This completes the proof of assertion (3).\\~\\ \textbf{$\mathcal{T}\phi$ belongs to $cl(\mathcal{D}_{t,s})$.}\\~\\
We first prove ${\mathcal{T}}_{1}$ belongs to $C^{1,1}(\mathcal{D})$ and hence also is in $cl(\mathcal{D}_{t,s})$.
Now, since $F$ is continuously differentiable, we have for $(t,p,i,s)\in {\mathcal{D}^{\mathrm{o}}}$,
\begin{align}
&\diffp{{\mathcal{T}}_{1}}{t}(t,p,i,s) = g(p)\frac{f(T-t+s)}{1-F(s)}, \label{d1t}\\
&\diffp{{\mathcal{T}}_{1}}{s}(t,p,i,s) = g(p)\left(-\frac{f(T-t+s)}{1-F(s)}+\frac{f(s)(1-F(T-t+s))}{{(1-F(s))}^2}\right).\label{d1s}
\end{align}
The continuous differentiability of $F$ and the continuity of $g$ again imply that $\diffp{{\mathcal{T}}_{1}}{t}$ and $\diffp{{\mathcal{T}}_{1}}{s}$ are continuous, hence ${\mathcal{T}}_1$ belongs to $C^{1,1}(\mathcal{D})$ and hence belongs to $cl(\mathcal{D}_{t,s})$ as well Moreover, $\mathcal{D}_{t,s}\mathcal{T}_1$ is given by 
\begin{equation}
\mathcal{D}_{t,s}\mathcal{T}_1(t,p,i,s)=g(p)\frac{f(s)(1-F(T-t+s))}{{(1-F(s))}^2}\label{d1ts}
\end{equation}\\
where $(t,p,i,s) \in {\mathcal{D}^{\mathrm{o}}}.$\\~\\
Next we prove ${\mathcal{T}}_{2}^{j}$  belongs to $C^{1,1}(\mathcal{D})$ and thus belongs to class $cl(\mathcal{D}_{t,s})$. For $(t,p,i,s)\in \mathcal{D}^{\mathrm{o}}$, $\varepsilon \in (-t, T-t), \; \varepsilon \neq 0$, and then using \eqref{c2t} we have
\begin{align*}
&\frac{{\mathcal{T}}_{2}^{j}(t+\varepsilon,p,i,s)-{\mathcal{T}}_{2}^{j}(t,p,i,s)}{\varepsilon} \\  &=\frac{1}{1-F(s)}\bigg(\int_{t+\varepsilon}^{T}\bigg(\frac{(fp_{ij})(v-t-\varepsilon+s)-(fp_{ij})(v-t+s)}{\varepsilon}\bigg)\phi(v,p(1+\delta\alpha(j)),j,0)dv\\
&-\frac{1}{\varepsilon}\int_{t}^{t+\varepsilon}(fp_{ij})(v-t+s)\phi(v,p(1+\delta\alpha(j)),j,0)dv\bigg).
\end{align*}
Proceeding similarly as in the case of proving continuity with respect to $t$ along with the application of Fundamental theorem of calculus for second term of the above equation, we get
\begin{align}
&\lim_{\varepsilon\rightarrow 0}\frac{{\mathcal{T}}_{2}^{j}(t+\varepsilon,p,i,s)-{\mathcal{T}}_{2}^{}(t,p,i,s)}{\varepsilon}=\nonumber\\
&-\frac{1}{1-F(s)}\int_{t}^{T}{(fp_{ij})}^{'}(v-t+s)\phi(v,p(1+\delta\alpha(j)),j,0)dv+\frac{(fp_{ij})(s)\phi(t,p(1+\delta\alpha(j)),j,0)}{1-F(s)}\label{d2t}.
\end{align}
Therefore, partial derivative of ${\mathcal{T}}_{2}^{j}$ with respect to $t$-variable exists on $\mathcal{D}^{\mathrm{o}}$ and it is given by \eqref{d2t}.\\~\\
Again for $(t,p,i,s)\in \mathcal{D}^{\mathrm{o}}$, $\varepsilon \in (-s, s), \; \varepsilon \neq 0$, using \eqref{c2s} we obtain
\begin{align*}
&\frac{{\mathcal{T}}_{2}^{j}(t,p,i,s+\varepsilon)-{\mathcal{T}}_{2}^{j}(t,p,i,s)}{\varepsilon}\\
&=\int_{t}^{T}\frac{1}{\varepsilon}\bigg(\frac{(fp_{ij})(v-t+s+\varepsilon)}{1-F(s+\varepsilon)}-\frac{(fp_{ij})(v-t+s)}{1-F(s)}\bigg)\phi(v,p(1+\delta\alpha(j)),j,0))dv.    
\end{align*}
Proceeding similar to the case of proving continuity for 
${\mathcal{T}}_{2}^{j}$ with respect to $s$, we infer
\begin{align}
&\lim_{\varepsilon\rightarrow 0}\frac{{\mathcal{T}}_{2}^{j}(t,p,i,s+\varepsilon)-{\mathcal{T}}_{2}^{}(t,p,i,s)}{\varepsilon}\nonumber\\ &=\int_{t}^{T}\bigg(\frac{{(fp_{ij})}^{'}(v-t+s)(1-F(s))+(fp_{ij})(v-t+s))f(s)}{{(1-F(s))}^2}\bigg)\phi(v,p(1+\delta\alpha(j)),j,0)dv.\label{d2s}
\end{align}
Therefore, partial derivative of ${\mathcal{T}}_{2}^{j}$ with respect to $s$ exists on $\mathcal{D}^{\mathrm{o}}$ and it is given by \eqref{d2s}.\\~\\
Further since $F, p_{ij}, f$ are continuously differentiable functions and $\phi$ is also continuous, proceeding as in the case of proving continuity for 
${\mathcal{T}}_{2}^{j}$, we conclude that $\diffp{{\mathcal{T}}_{2}^{j}}{t},\diffp{{\mathcal{T}}_{2}^{j}}{s} \in C{(\mathcal{D}}^{\mathrm{o}})$. Hence it has been proved ${\mathcal{T}}_{2}^{j} \in C^{1,1}(\mathcal{D})$. Thus it also lies in class $cl(\mathcal{D}_{t,s})$ and using \eqref{d2t} and \eqref{d2s}, we see that for $(t,p,i,s) \in {\mathcal{D}^{\mathrm{o}}}$,
\begin{align}
\mathcal{D}_{t,s}\mathcal{T}_2^j(t,p,i,s)&=-\frac{(fp_{ij})(s)\phi(t,p(1+\delta\alpha(j)),j,0)}{1-F(s)}\nonumber\\
&-\int_{t}^{T}\frac{(fp_{ij})(v-t+s))f(s)}{{(1-F(s))}^2}\phi(v,p(1+\delta\alpha(j)),j,0)dv.\label{d2ts}   
\end{align}
We further prove ${\mathcal{T}}_{3}$ belongs to class $cl(\mathcal{D}_{t,s})$. For this consider $(t,p,i,s) \in {\mathcal{D}}^{\mathrm{o}}$, $\varepsilon \in (-s,s)\cap(-t,T-t), \varepsilon \neq 0,$
\begin{align*}
&\frac{\mathcal{T}_3(t+\varepsilon,p,i,s+\varepsilon)-\mathcal{T}_3(t,p,i,s)}{\varepsilon}=\frac{1}{\varepsilon}\bigg(\int_{t+\varepsilon}^{T}\frac{1-F(v-t+s)}{1-F(s+\varepsilon)}w(v,p,i,v-t+s)dv\\
&-\int_{t}^{T}\frac{1-F(v-t+s)}{1-F(s)}w(v,p,i,v-t+s)dv\bigg)\\
&=\int_{t+\varepsilon}^{T}\frac{1}{\varepsilon}\left(\frac{1}{1-F(s+\varepsilon)}-\frac{1}{1-F(s)}\right)(1-F(v-t+s))w(v,p,i,v-t+s)dv\\
&-\frac{1}{\varepsilon}\int_{t}^{t+\varepsilon}\frac{1-F(v-t+s)}{1-F(s)}w(v,p,i,v-t+s)dv.
\end{align*}
Since F is differentiable, $w(\cdot,p,i,\cdot)$ is continuous and the integration is over a compact set, using bounded convergence theorem and Fundamental theorem of calculus, we have
\begin{align}
&\lim_{\varepsilon\rightarrow 0}\frac{\mathcal{T}_3(t+\varepsilon,p,i,s+\varepsilon)-\mathcal{T}_3(t,p,i,s)}{\varepsilon}\nonumber\\
&=\int_{t}^{T} \frac{f(s)(1-F(v-t+s)}{{(1-F(s))}^2}w(v,p,i,v-t+s)dv - w(t,p,i,s).\label{d3ts}
\end{align}
Therefore, $\mathcal{D}_{t,s}\mathcal{T}_3$ exists and it is given by \eqref{d3ts}. We also note that the continuity of $f,F$ and $w$ implies that $\mathcal{D}_{t,s}\mathcal{T}_3$ is continuous and thus, $\mathcal{T}_3 \in cl(\mathcal{D}_{t,s})$.\par
Since we have proved $\mathcal{T}_3,\mathcal{T}_2^j$ and $\mathcal{T}_3$ belong to $cl(\mathcal{D}_{t,s})$, therefore using \eqref{T1T2T3} we conclude $\mathcal{T}\phi \in cl(\mathcal{D}_{t,s})$, thus proving assertion (4).\\~\\
Now for proving assertion (5), we assume $w_s:= \diffp{w}{s}$ belongs to $C({\mathcal{D}}^{\mathrm{o}})$. Let $(t,p,i,s)\in \mathcal{D}^{\mathrm{o}}$, $\varepsilon \in (-t, T-t), \; \varepsilon \neq 0$, then
\begin{align*}
&\frac{\mathcal{T}_3(t+\varepsilon,p,i,s)-\mathcal{T}_3(t,p,i,s)}{\varepsilon}=\frac{1}{\varepsilon}\bigg(\int_{t+\varepsilon}^{T}\frac{1-F(v-t-\varepsilon+s)}{1-F(s)}w(v,p,i,v-t-\varepsilon+s)dv\\
&-\int_{t}^{T}\frac{1-F(v-t+s)}{1-F(s)}w(v,p,i,v-t+s)dv\bigg)\\   
&=\frac{1}{1-F(s)}\bigg(\int_{t+\varepsilon}^{T}\frac{1}{\varepsilon}\bigg((1-F(v-t-\varepsilon+s))w(v,p,i,v-t-\varepsilon+s)-\\
&(1-F(v-t+s))w(v,p,i,v-t+s)dv\bigg)-\frac{1}{\varepsilon}\int_{t}^{t+\varepsilon}(1-F(v-t+s))w(v,p,i,v-t+s)dv\bigg).
\end{align*}
Since $F$ is differentiable and we have assumed the existence of $\diffp{w}{s}$ in this case, then using the similar arguments as in the case of proving the existence of $\diffp{\mathcal{T}_2^j}{t}$, we obtain
\begin{align}
&\lim_{\varepsilon\rightarrow 0}\frac{\mathcal{T}_3(t+\varepsilon,p,i,s)-\mathcal{T}_3(t,p,i,s)}{\varepsilon}=\nonumber\\
&\int_{t}^{T}\frac{f(v-t+s)w(v,p,i,v-t+s)-(1-F(v-t+s))w_s(v,p,i,v-t+s)}{1-F(s)}dv-w(t,p,i,s).\label{d3t}
\end{align}
Therefore $\diffp{\mathcal{T}_3}{t}$ exists on $\mathcal{D}^{\mathrm{o}}$ if $w_s$ exists on $\mathcal{D}^{\mathrm{o}}$ and it is given by \eqref{d3t}. Now again for $(t,p,i,s)\in \mathcal{D}^{\mathrm{o}}$, $\varepsilon \in (-s, s), \; \varepsilon \neq 0$, consider 
\begin{align*}
\frac{\mathcal{T}_3(t,p,i,s+\varepsilon)-\mathcal{T}_3(t,p,i,s)}{\varepsilon}&=\int_{t}^{T}\frac{1}{\varepsilon}\bigg(\frac{(1-F(v-t+s+\varepsilon))w(v,p,i,v-t+s+\varepsilon)}{1-F(s+\varepsilon)}\\
&-\frac{(1-F(v-t+s))w(v,p,i,v-t+s)}{1-F(s)}\bigg)dv. 
\end{align*}
Again since $F$ is differentiable and $w_s$ exists, we have
\begin{align}
&\lim_{\varepsilon\rightarrow 0}\frac{\mathcal{T}_3(t,p,i,s+\varepsilon)-\mathcal{T}_3(t,p,i,s)}{\varepsilon}\nonumber\\
&=\int_{t}^{T}\bigg(\frac{(1-F(v-t+s))w_s(v,p,i,v-t+s)-f(v-t+s)w(v,p,i,v-t+s)}{1-F(s)}\nonumber\\
&\qquad+\frac{f(s)F(v-t+s))w(v,p,i,v-t+s)}{{(1-F(s))}^2}\bigg)dv. \label{d3s}
\end{align}
Hence $\diffp{\mathcal{T}_3}{s}$ exists on $\mathcal{D}^{\mathrm{o}}$ and it is given by \eqref{d3s}.
Moreover, the continuity of $f,w$ and $w_s$ ensures that $\diffp{\mathcal{T}_3}{t}$ and $\diffp{\mathcal{T}_3}{t}$ are continuous, using the same arguments as we had in the case of proving continuity of $\mathcal{T}_2^j$. Thus, $\mathcal{T}_3 \in C^{1,1}\mathcal{(D)}.$\\
As we have proved $\mathcal{T}_1, \mathcal{T}_2^j$ and $\mathcal{T}_3$ (if we assume $w_s$ exists and is continuous), therefore $\mathcal{T}\phi \in C^{1,1}\mathcal{(D)}$ in this case. This completes the proof of assertion (5).
\end{proof}
\vspace{3mm}
\begin{theorem}\label{thm2}
Consider the following partial differential equation with the terminal condition as follows.
\begin{align}
\diffp{\phi}{t}+\diffp{\phi}{s}+\sum_{j\not \equiv i}h_{ij}(s)(\phi(t,&p(1+\delta \alpha(j)),j,0)-\phi(t,p,i,s))+w(t,p,i,s) = 0 \label{thm2eq1}\\
&\phi(T,p,i,s)=g(p) \label{thm2eq2}
\end{align}
where $(t,p,i,s) \in {\mathcal{D}^{\mathrm{o}}}$ and $w$, $g$ are as in Theorem 3.1. Then the equation \eqref{thm2eq1} with the terminal condition \eqref{thm2eq2} has a unique solution in the class of functions $V \cap cl(\mathcal{D}_{t,s})$. Further if $w$ has continuous partial derivative with respect to $s$ then the above said solution belongs to $V\cap C^{1,1}\mathcal{(D)}$.
\end{theorem}
\begin{proof}
We prove the existence of a solution of the equation \eqref{thm2eq1} with the terminal condition \eqref{thm2eq2} by showing that the solution of the equation \eqref{thm1eq2} belonging to $V \cap cl(D_{t,s})$ also satisfies \eqref{thm2eq1}-\eqref{thm2eq2}. Let us assume $\phi$ to be the unique solution of \eqref{thm1eq2} in $V$ which is given by \eqref{thm1eq1}. Then using \eqref{T1T2T3}, we have for $(t,p,i,s) \in {\mathcal{D}^{\mathrm{o}}}$,
\begin{align*}
&\diffp{\phi}{t}+\diffp{\phi}{s}=\mathcal{D}_{t,s}\left(\mathcal{T}_1+\sum_{j \not\equiv i}{\mathcal{T}}_{2}^j+{\mathcal{T}}_{3}\right)(t,p,i,s).
\end{align*}
Using \eqref{d1ts}, \eqref{d2ts} and \eqref{d3ts}, we obtain
\begin{align*}
&\diffp{\phi}{t}+\diffp{\phi}{s}=g(p)\frac{f(s)(1-F(T-t+s))}{{(1-F(s))}^2}\\
&-\sum_{j \not\equiv i}\bigg(\frac{(fp_{ij})(s)\phi(t,p(1+\delta\alpha(j)),j,0)}{1-F(s)}+\int_{t}^{T}\frac{(fp_{ij})(v-t+s))f(s)}{{(1-F(s))}^2}\phi(v,p(1+\delta\alpha(j)),j,0)dv\bigg)\\
&+\int_{t}^{T} \frac{f(s)(1-F(v-t+s)}{{(1-f(s))}^2}w(v,p,i,v-t+s)dv-w(t,p,i,s)\\
&=\frac{f(s)}{1-F(s)}\bigg(\frac{g(p)(1-F(T-t+s))}{1-F(s)}\\\
&+\int_{t}^{T}\frac{f(v-t+s)}{1-F(s)}\sum_{j \not\equiv i}p_{ij}(v-t+s)\phi(v,p(1+\delta\alpha(j)),j,0)dv\\
&+\int_{t}^{T} \frac{1-F(v-t+s}{1-F(s)}w(v,p,i,v-t+s)dv\bigg)-\sum_{j \not\equiv i}\frac{f(s)p_{ij}(s)}{1-F(s)}\phi(t,p(1+\delta\alpha(j)),j,0)-w(t,p,i,s)\\
&=\frac{f(s)}{1-F(s)}\mathcal{T}\phi(t,p,i,s)-\sum_{j \not\equiv i}\frac{f(s)p_{ij}(s)}{1-F(s)}\phi(t,p(1+\delta\alpha(j)),j,0)-w(t,p,i,s).
\end{align*}
Now using \eqref{hij}, \eqref{thm1eq2} and the fact that $\sum_{j \not\equiv i}p_{ij}(s)=1$, we have
\begin{align*}
&\diffp{\phi}{t}+\diffp{\phi}{s}+w(t,p,i,s)= \sum_{j \not\equiv i}h_{ij}(s)(\phi(t,p,i,s)-\phi(t,p(1+\delta\alpha(j)),j,0)).
\end{align*}
Also by plugging $t=T$ in \eqref{thm1eq1}, we get $\phi(T,p,i,s) = g(p)$. Since $\phi$ belongs to $V \cap cl(\mathcal{D}_{t,s})$, we conclude the partial differential equation \eqref{thm2eq1} with the terminal condition \eqref{thm2eq2} has a solution in the class of functions $V \cap cl(\mathcal{D}_{t,s})$.\\~\\
Now we prove the uniqueness of the solution. Let $\phi \in V \cap cl(D_{t,s})$ be a solution of \eqref{thm2eq1}-\eqref{thm2eq2}.\\
Then consider the process
\[N_t^\phi:= \phi(t,P_t,I_t,S_t)+\int_{0}^{t}w(v,P_{v^-},I_{v^-},S_{v^-})dv.\] 
Similar to \eqref{gpsi}, we define 
\begin{align*}
G^\phi(v,z):=&\phi(v,P_{{v}^{-}}+g_3(I_{{v}^{-}},S_{{v}^{-}},z),I_{{v}^{-}}+g_2(I_{{v}^{-}},S_{{v}^{-}},z),S_{{v}^{-}} + g_1(S_{{v}^{-}},I_{{v}^{-}},z))\\
&-\phi(v,P_{{v}^{-}},I_{{v}^{-}},S_{{v}^{-}}).
\end{align*}
Then proceeding as in \eqref{genpis1}-\eqref{genpis2}, we get
\begin{equation}
\int_{0}^{t}\int_{\mathbb{R}}G^\phi(v,z)dz\;dv=\int_{0}^{t} \sum_{j\not \equiv I_{{v}^{-}}}h_{I_{v^{-}}j}(S_{{v}^{-}})(\phi(v,P_{{v}^{-}}(1+\delta\alpha(j)),j,0)-\phi(v,P_{{v}^{-}},I_{{v}^{-}},S_{{v}^{-}}))dv. \label{integralgphi} 
\end{equation}
Since $\phi$ and $w$ are continuous and $\phi$ belongs to $\mathcal{D}_{t,s}$, using It\^{o}'s formula, we obtain
\begin{align*}
dN_t^\phi=&\bigg(w(t,P_{t^-},I_{t^-},S_{t^-})+\left(\diffp{}{t}+\diffp{}{s}\right)\phi(t,P_{t^-},I_{t^-},S_{t^-})+\\
&\sum_{j\neq I_{{t}^{-}}}h_{I_{{t}^{-}j}}(S_{{t}^{-}})(\phi(t,P_{{t}^{-}}(1+\delta\alpha(j)),j,0)-\phi(t,P_{{t}^{-}},I_{{t}^{-}},S_{{t}^{-}}))\bigg)dt+dM_t,
\end{align*}
where $M:=\{M_t\}_{t\ge 0}$ is given by the stochastic integral $M_t := \int_{0}^{t} \int_{\mathbb{R}}G^\phi(v,z)\Tilde{\wp}(dz,dv)$ for all $t\ge 0$.
The equation \eqref{thm2eq1} implies that $N_t^\phi=M_t+\phi(0,P_0,I_0,S_0).$ \\\\
Now since $\phi \in V$, $
\phi(v,p,i,s)\leq ||\phi||_V(1+p)$ for all $(v,p,i,s) \in \mathcal{D}.$\\
It follows from the above assertion, the equation \eqref{integralgphi} and the assumption (A2) that 
\begin{equation}
\int_{0}^{t}\int_{\mathbb{R}}|G^\phi(v,z)|dz\;dv \leq 4c_1\|\phi\|_V t+\int_{0}^{t}P_{v^-} dv. \label{Gphi}
\end{equation}
Let $T_n$ denote the time of $n\textsuperscript{th}$ jump of the process $(P,I,S)$. Let for $t>0$, $N(t)$ denotes the number of jumps of the process $(P,I,S)$ up to time `$t$'. Again we observe that the process $P$ can be written as 
\begin{equation}
P_t = P_0\prod_{k=1}^{N(t)}(1+\delta\alpha(I_{T_k}))\text{ a.s..}\label{Pt}   
\end{equation}
We further note that due to assumption (A2), the Poisson mass points that correspond to a jump in the processes $P,I$ and $S$ lie in the rectangle $[0,T] \times [0,2c_1]$. We denote the number of Poisson mass points lying in the above said rectangle by $N_{\wp}(T)$. Since the intensity measure of $\wp$ i.e., the product Lebesgue measure of the rectangle $[0,T] \times [0,2c_1]$ is finite, hence $N_{\wp}(T)<\infty$ almost surely. Again as $N(t) \leq N_{\wp}(t)$, thus using \eqref{Pt} we have
\[P_{v^-} \leq P_0(1+\delta)^{N_{\wp}(T)} \text{ for every } v \in (0,T] \text{ a.s.}.\]
Now for $\delta>0$ we calculate $E[(1+\delta)^{N_{\wp}(T)}]$. Since $N_{\wp}(T)=\wp([0,T]\times[0,2c_1])$ is a Poisson random variable, we have
\begin{align}
E[(1+\delta)^{N_{\wp}(T)}]&=\sum_{n=0}^{\infty}(1+\delta)^n\mathbb{P}(N_{\wp}(T)=n)
=e^{2c_1T\delta} <\infty. \label{EbNpT1}
\end{align}
Now using the equation \eqref{Gphi} we have
\[\int_{0}^{t}\int_{\mathbb{R}}|G^\phi(v,z)|dz\;dv\leq t(4c_1||\phi||_V+P_0(1+\delta)^{N_{\wp}(T)}) \text{ a.s..}\]
Therefore, using \eqref{EbNpT1} and the assumption (A5) we have
\begin{equation}
E\bigg[\int_{0}^{t}\int_{\mathbb{R}} |G^\phi(v,z)|dz\;dv\bigg]<\infty\;\;\text{for each } t\in[0,\infty). \label{gphiex}   
\end{equation}
Note that $G^\phi = \max\{G^\phi,0\}-\max\{-G^\phi,0\}$. It follows from \eqref{gphiex} that
\[E\bigg[\int_{0}^{t}\int_{\mathbb{R}}\max\{G^\phi,0\} dz\;dv\bigg]<\infty\;\;\text{for each } t\in[0,\infty), \;\; \text{and}\]
\[E\bigg[\int_{0}^{t}\int_{\mathbb{R}}\max\{-G^\phi,0\} dz\;dv\bigg]<\infty\;\;\text{for each } t\in[0,\infty).\]
Since $G^\phi$ is a predictable process as it is a left continuous process, hence $\max\{G^\phi,0\}$ and 
$\max\{-G^\phi,0\}$ are positive predictable processes.\par
Therefore, $M^+:=\{M^{+}_t\}_{t\ge 0}$  and $M^-_t :=\{M^{-}_t\}_{t\ge 0}$ where $M^{+}_t = \int_{0}^{t} \int_{\mathbb{R}}\max\{G^\phi,0\}\Tilde{\wp}(dz,dv)$, $M^{-}_t=\int_{0}^{t} \int_{\mathbb{R}}\max\{-G^\phi,0\}\Tilde{\wp}(dz,dv)$ are martingales with respect to the filtration $\{{\mathscr{F}}_t\}$ (see Remark 6.4c (p-300) of \cite{cinlar2011probability}). Hence $M$ is also a martingale with respect to $\{{\mathscr{F}}_t\}$.\\~\\ 
Moreover, since the process  $N^\phi=M+\phi(0,P_0,I_0,S_0)$ is also a martingale using the above assertion, we have
\begin{align*}
&N_t^\phi=E[N_T^\phi|{\mathscr{F}}_t].\\
\text{i.e., }&\phi(t,P_t,I_t,S_t)=E\bigg[g(P_T)+\int_{t}^{T}w(v,P_{v^-},I_{v^-},S_{v^-})dv \bigg\vert {\mathscr{F}}_t\bigg]. 
\end{align*}
Using Markovity of the process $(P,I,S)$ we see that the above equation can be written as
\begin{equation}
\phi(t,P_t,I_t,S_t)=E\bigg[g(P_T)+\int_{t}^{T}w(v,P_{v^-},I_{v^-},S_{v^-})dv \bigg\vert P_t,I_t,S_t\bigg] \label{storep1}  
\end{equation}
almost surely for every $t\in[0,T]$.\\
Furthermore, the assumption (A6) ascertains that the distribution of $(P_t,I_t,S_t)$ has full support for every $t \in [0,T]$ (see lemma \ref{lemmaA1}) then evaluating \eqref{storep1} for $w \in \{P_t=p,I_t=i,S_t=s\}$ we obtain
\begin{equation}
\phi(t,p,i,s)=E\bigg[g(P_T)+\int_{t}^{T}w(v,P_{v^-},I_{v^-},S_{v^-})dv \bigg\vert P_t=p,I_t=i,S_t=s\bigg]\label{storep}
\end{equation}
where $(t,p,i,s) \in \mathcal{D}$.\\\\
We clearly see from \eqref{storep} that every solution has the above stochastic representation therefore we conclude that the solution $\phi$ is unique. Hence the equation \eqref{thm2eq1} with the terminal condition \eqref{thm2eq2} has a unique solution in the class of functions $V \cap cl(\mathcal{D}_{t,s})$.\\~\\
Further, if $w_s$ exists and it is continuous then according to the assertion (5) of theorem \ref{thm1}, the solution $\phi$ of \eqref{thm1eq1}-\eqref{thm1eq2} belongs to $V\cap C^{1,1}\mathcal{(D)}$. Since $\phi$ also satisfies the equations \eqref{thm2eq1}-\eqref{thm2eq2} and the solution to these equations is unique in the class $V \cap cl(\mathcal{D}_{t,s})$, therefore \eqref{thm2eq1}-\eqref{thm2eq2} have a unique solution in the class $V\cap C^{1,1}\mathcal{(D)}$.
\end{proof}
\begin{remark}\label{remark3.1}
Let $\pi(t,p,i,s):=E[P_T\mid P_t=p,I_t=i,S_t=s$ (where $(t,p,i,s)\in\mathcal{D}$) denote the conditional expectation of the stock price at the terminal time given the event when the process is in state `$i$' with age `$s$' and the stock price is `$p$'. Studying $\pi$ would be helpful in gaining more insight into the dynamics of the stock price and understanding the behaviour of stock price over time. In \cite{pham2}, it has been shown that $\pi$ is the unique viscosity solution to a linear partial differential equation with a terminal condition. We prove below that $\pi$ belongs to $C^{1,1}\mathcal{(D)}$ and it uniquely satisfies a similar terminal-value problem in the classical sense.
\end{remark}
\begin{corollary}\label{corollary3.1}
$\pi$ belongs to $V\cap C^{1,1}\mathcal{(D)}$ and it is the unique solution to the following partial differential equation
\begin{equation}
\diffp{\pi}{t}+\diffp{\pi}{s}+\sum_{j\not \equiv i}h_{ij}(s)(\pi(t,p(1+\delta \alpha(j)),j,0)-\pi(t,p,i,s))=0 \label{pi}
\end{equation}
with the terminal condition
\begin{equation}
\pi(T,p,i,s)=p. \label{piT}    
\end{equation}
where $t\in[0,T],p\in[0,\infty),i\in\mathcal{X}$ and $s\in[0,\infty)$.
\end{corollary}
\begin{proof}
We consider the case when $g$ is the identity map on $[0,\infty)$ and $w=0$ in the equations \eqref{thm2eq1} and \eqref{thm2eq2}. Then using \eqref{storep} we see that the unique solution of \eqref{thm2eq1}-\eqref{thm2eq2} in the class $V \cap cl(D_{t,s})$ is nothing but $\pi$. Moreover since $w=0$, theorem \ref{thm2} asserts that the solution of \eqref{thm2eq1}-\eqref{thm2eq2} belongs to $V\cap C^{1,1}\mathcal{(D)}$. Hence $\pi$ belongs to $V\cap C^{1,1}\mathcal{(D)}$.     
\end{proof}

\section{The Market-making strategy and the portfolio of the agent}
\subsection{Formulation of the model}The incoming market orders can be categorised into two kinds, small market orders and big market orders depending on the size of the order. The `small' market orders do not cause any change in the best bid/ask price whereas the big market orders lead to a jump in the best bid/ask price that can be upwards or downwards depending on the situation. The outstanding limit orders are documented in a Limit Order Book (LOB) which is maintained by the specialist at the exchange. We are analysing the market-making problem from the point of view of an agent whose strategy is to continuously place limit orders of constant size `$L$' on both side of the LOB at the best price. We assume that the market order dynamics is not affected by the presence of the agent and the incoming orders get executed by the agent only on one side of the LOB at a time.\\ \par
Now, let $\vartheta_{+}$ and $\vartheta_{-}$ be two probability measures on $\mathcal{K}:=\{0,1,\ldots,K\}$, first $K$ number of non-negative integers. For each $\nu \in \{+,-\}$, let $\lambda_{\nu} : [0,\infty) \rightarrow [0, \infty)$ be a measurable function. Now for each $s\geq 0$, we define a family of right open, left closed non-overlapping intervals $\{\Lambda_{\nu k}(s) \mid \nu \in \{+,-\}, k \in \mathcal{K}\}$ such that this collection of intervals is disjoint from $H(s)$ and the length of each interval $\Lambda_{\nu k}(s)$ is equal to $\lambda_{\nu}(s)\vartheta_{\nu}(\{k\})$. Let $\epsilon > 0$ be the fixed cost associated with each transaction. We also make the following assumptions:
\begin{itemize}
\item [(A7)]$\sup\limits_{\nu \in \{+,-\},y\in [0,\infty)} \lambda_{\nu}(y)<\infty$ and we denote it by $c_2$.
\item[(A8)]${\lambda}_+$ and ${\lambda}_-$ are continuous functions.
\end{itemize}~\\\par
\textbf{Financial interpretation.} Here $\nu$ represents the two sides of the LOB where by convention, we assume that $\nu=-$ represents the bid side and $\nu=+$ represents the ask side of the LOB. Also, $\vartheta_{\nu}$ gives the probability distribution of the random number of units of the stock getting traded on the side $\nu$ of the LOB, and $k \in\mathcal{K}$ denotes the number of quantities of the stock that get executed by the agent. The arrival of and execution of a small market order on side $\nu$ with `$k$' units being traded by the agent is captured by the interval $\Lambda_{\nu k}(s)$ where `$s$' denotes the age of the process. The intensity of incoming small market orders is given by ${\lambda}_{\nu}(s)$. \\\par
On the other hand, in the case of big market orders, the kind of order that must have arrived or the side that gets executed is determined by the current and the last price jump direction.The table below demonstrates this in different scenarios. Moreover, the interval $H_{ij}(s)$ corresponds to the case when a big market order arrives causing a jump in the stock price with the state process changing its state to `$j$' from `$i$ and `$s$' being the time elapsed since last jump, with the intensity of the arrival of big market orders equal to $h_{ij}(s)$. The disjointness of these intervals ensures that only one kind of order is being executed on one side of the LOB at a time. 
\begin{table}[h]
\centering
\caption{Demonstrating the side of LOB that gets executed upon the arrival of a big market order}
\begin{tabular}{|c|c|c|c|c|P{6cm}|}
\hline
$I_{t^-}$ & $\alpha(I_{t^-})$ & $I_t$ & $\alpha(I_t)$ & Type of incoming order & Side of LOB pertaining to executed order\\
\hline
1 & -1 & 3 & -1 & sell & Bid(-)\\
\hline
1 & -1 & 4 & +1 & buy & Ask(+)\\
\hline
2 & +1 & 3 & -1 & sell & Bid(-)\\
\hline
2 & +1 & 4 & +1 & buy & Ask(+)\\
\hline
3 & -1 & 1 & -1 & sell & Bid(-)\\
\hline
3 & -1 & 2 & +1 & buy & Ask(+)\\
\hline
4 & +1 & 1 & -1 & sell & Bid(-)\\
\hline
4 & +1 & 2 & +1 & buy & Ask(+)\\
\hline
\end{tabular}
\label{fig:side}
\end{table}\\
\textbf{The wealth and inventory processes.} The strategy adopted by the agent is described using a stochastic process $\Hat{l}=({\Hat{l}}^+,{\Hat{l}}^-)$ where ${\Hat{l}}^+$ and ${\Hat{l}}^-$ are predictable processes taking values in the control space $\{0,1\}\times\{0,1\}$ denoted by $L$. We also recall that the bid-ask spread was assumed to be equal to $2p\delta$ where $p$ is the mid-price. Now for $l=(l^+,l^-) \in L$, let $g_4$ and $g_5$ be two real valued maps defined on the sets $\mathcal{X}\times[0,\infty)\times[0,\infty)\times\mathbb{R}$ and $\mathcal{X}\times[0,\infty)\times \mathbb{R}$ respectively as follows:
\begin{align}
g_4(i,s,p,l,z) &= \sum_{\nu=\pm,k\in\mathcal{K}}\nu l^\nu k( p(1+\nu\delta)-\nu\epsilon)1_{\Lambda_{\nu k}(s)}(z)\nonumber\\
&+\sum_{j \not\equiv i}\alpha(j)l^{\bar{\alpha}(j)}K(p(1+\delta\alpha(j))-\alpha(j)\epsilon)1_{H_{ij}(s)}(z), \label{g4} \\
g_5(i,s,l,z) &= \sum_{\nu=\pm,k\in\mathcal{K}} \nu l^\nu k 1_{\Lambda_{\nu k}(s)}(z) + \sum_{j \not\equiv i}\alpha(j)l^{\bar{\alpha}(j)}K1_{H_{ij}(s)}(z). \label{g5}
\end{align}
The table below describes the change in portfolio of the agent when an incoming market order arrives, subject to the control $l$, thus explaining the rationale behind the definition of functions $g_4$ and $g_5$.
\begin{table}[h]
\caption{Depicting change in the portfolio of the agent corresponding to the control `$l$' upon the arrival of a market order in different cases}
\centering
\begin{tabular}{|P{1.8cm}|P{1.5cm}|P{1.3cm}|P{2cm}|P{3.2cm}|P{1.8cm}|P{2.1cm}|}
\hline
Type of incoming market order & Side of LOB & No of units traded & Price at which order gets executed & Change in Wealth & Change in Inventory & Change in Portfolio\\
\hline
Small & Ask(+) & $k$ & $p(1+\delta)$ & $l^+ k(p(1+\delta)-\epsilon)$ & $-l^+ k$ & $l^+k(p\delta-\epsilon)$\\
\hline
Small & Bid(-) & $k$ & $p(1-\delta)$ & $l^- k(-p(1-\delta)-\epsilon)$ & $l^- k$ & $l^- k(p\delta-\epsilon)$\\
\hline
Big & Ask(+) & $K$ & $p(1+\delta)$ & $l^+ K(p(1+\delta)-\epsilon)$ & $-l^+ K$ & $-l^+ K\epsilon$\\
\hline
Big & Bid(-) & $K$ & $p(1-\delta)$ & $l^- K(-p(1-\delta)-\epsilon)$ & $l^- K$ & $-l^- K\epsilon$\\
\hline
\end{tabular}
\label{fig:portfolio}
\end{table}\\
We now describe the portfolio of the agent associated with the control process $\Hat{l}$ using the processes $X$ and $Y$ which denote the wealth and the inventory of the agent respectively and are given by
\begin{align}
X_t &=X_0 + \int_{(0,t]\times\mathbb{R}}g_4(I_{s^-},S_{s^-},P_{s^-},{\Hat{l}}_s,z)\wp(ds,dz), \label{wealthfn}\\
Y_t &=Y_0-\int_{(0,t]\times\mathbb{R}}g_5(I_{s^-},S_{s^-},{\Hat{l}}_s,z)\wp(ds,dz), \label{inventoryfn}
\end{align}
where $X_0$ and $Y_0$ are ${\mathscr{F}}_0$-valued random variables valued in $\mathbb{R}$ and $\mathbb{Z}$ respectively.
\subsection{The infinitesimal generator of the control process}
In the case of a constant control taking value $l$ $(=(l^+,l^-))$ $\in L$, we observe that due to the assumptions (A2) and (A7), there exists a pathwise unique strong solution $(P,I,S,X,Y) :=\{(P_t, I_t, S_t,X_t,Y_t)\}_{t\ge 0}$ to the system of equations \eqref{agefn}-\eqref{pricefn},\eqref{wealthfn}-\eqref{inventoryfn} that is right continuous having left limits, locally bounded and strongly Markov (see Theorem 3.4 (p-474) of \cite{cinlar2011probability}) similar to the case of the system of equations \eqref{agefn}-\eqref{pricefn}.\\ \par
We now find the infinitesimal generator of the process $(P,I,S,X,Y)$ in the above mentioned case of a constant control taking value $l$ $(=(l^+,l^-))$ $\in L$. To this end, consider ${\psi}$ be a real-valued continuous map defined on the set $[0,\infty)\times\mathcal{X}\times[0,\infty)\times\mathbb{R}\times\mathbb{Z}$ which is compactly supported and continuously differentiable function in the third variable. Let $\Tilde{G}^{{\psi}}$ be real valued map on $(0,T]\times\mathbb{R}$ such that for $(v,z)\in (0,T]\times\mathbb{R}$,
\begin{align}
\Tilde{G}^{{\psi}}(v,z):=&\psi(P_{{v}^{-}}(1+g_3(I_{{v}^{-}},S_{{v}^{-}},z)),I_{{v}^{-}}+g_2(I_{{v}^{-}},S_{{v}^{-}},z),S_{{v}^{-}}-g_1(S_{{v}^{-}},I_{{v}^{-}},z),\nonumber\\
&X_{{v}^{-}}+g_4(I_{{v}^{-}},S_{{v}^{-}},P_{{v}^{-}},l,z),Y_{{v}^{-}}-g_5(I_{{v}^{-}},S_{{v}^{-}},l,z))-{\psi}(P_{{v}^{-}},I_{{v}^{-}},S_{{v}^{-}},X_{{v}^{-}},Y_{{v}^{-}}). \label{gtildepsi}
\end{align}
Using It\^{o}'s formula we have,
\begin{align}
&\psi(P_t,I_t,S_t,X_t,Y_t) - \psi(P_0,I_0,S_0,X_0,Y_0) \nonumber\\
&= \int_{0}^{t}\diffp{\psi}{s}(P_{v^-},I_{v^-},S_{v^-},X_{v^-},Y_{v^-})dv+\int_{0}^{t} \int_{\mathbb{R}}\tilde{G}^{\psi}(v,z)\wp(dz,dv).\label{genpisxy1}
\end{align}
Now,
\begin{align*}
\int_{0}^{t} \int_{\mathbb{R}}\Tilde{G}^{\psi}(v,z)dz\;dv\\
=\int_{0}^{t}\int_{\mathbb{R}}\bigg(\psi\bigg(&P_{{v}^{-}}\bigg(1 + \sum_{j \not\equiv I_{{v}^{-}}} \delta \alpha(j)1_{H_{I_{v^{-}}j}(S_{{v}^{-}})}(z)\bigg),I_{{v}^{-}}+\sum_{j \not\equiv I_{{v}^{-}}} (j-I_{{v}^{-}})1_{H_{I_{v^{-}}j}(S_{{v}^{-}})}(z),\\
&S_{{v}^{-}}(1-1_{H(S_{{v}^{-}})}(z)),X_{{v}^{-}}+\sum_{\nu=\pm,k\in\mathcal{K}}\nu l^\nu k( P_{{v}^{-}}(1+\nu\delta)-\nu\epsilon)1_{\Lambda_{\nu k}(S_{{v}^{-}})}(z)\\
&+\sum_{j \not\equiv I_{{v}^{-}}}\alpha(j)l^{\bar{\alpha}(j)}K(P_{{v}^{-}}(1+\delta\alpha(j))-\alpha(j)\epsilon)1_{H_{I_{{v}^{-}}j}(S_{{v}^{-}})}(z),\\
&Y_{{v}^{-}}-\sum_{\nu=\pm,k\in\mathcal{K}} \nu l^\nu k 1_{\Lambda_{\nu k}(S_{{v}^{-}})}(z)+\sum_{j \not\equiv I_{{v}^{-}}}\alpha(j) l^{\bar{\alpha}(j)}K1_{H_{I_{{v}^{-}}j}(S_{{v}^{-}})}(z)\\
-\psi(&P_{{v}^{-}},I_{{v}^{-}},S_{{v}^{-}},X_{{v}^{-}},Y_{{v}^{-}})\bigg)dz\;dv\\
\end{align*}
\begin{align}
=\int_0^t \bigg(&\sum_{\nu=\pm,k\in\mathcal{K}} \lambda_{\nu}(S_{{v}^{-}})\vartheta_{\nu}(\{k\})(\psi(P_{{v}^{-}},I_{{v}^{-}},S_{{v}^{-}},X_{{v}^{-}}+\nu l^\nu k( P_{{v}^{-}}(1+\nu\delta)-\nu\epsilon),Y_{{v}^{-}}-\nu l^\nu k)\nonumber\\
-&\psi(P_{{v}^{-}},I_{{v}^{-}},S_{{v}^{-}},X_{{v}^{-}},Y_{{v}^{-}})\bigg)dv\nonumber\\
+\int_0^t\bigg(&\sum_{j\neq I_{{v}^{-}}}h_{I_{{v}^{-}}j}(S_{{v}^{-}})(\psi(P_{{v}^{-}}(1+\delta\alpha(j)),j,0,X_{{v}^{-}}+\alpha(j) l^{\bar{\alpha}(j)}K(P_{{v}^{-}}(1+\delta\alpha(j))-\alpha(j)\epsilon),\nonumber\\
&Y_{{v}^{-}}-\alpha(j)l^{\bar{\alpha}(j)}K)-\psi(P_{{v}^{-}},I_{{v}^{-}},S_{{v}^{-}},X_{{v}^{-}},Y_{{v}^{-}})\bigg)dv. \label{genpisxy2}
\end{align}
Hence using \eqref{genpisxy2} we see that $\int_{0}^{t} \int_{\mathbb{R}}\Tilde{G}^\psi(v,z)dz\;dv$ is finite due to the assumptions (A2), (A7) and the fact that $\psi$ is compactly supported. So we can write
\begin{equation}
\int_{0}^{t} \int_{\mathbb{R}}\Tilde{G}^\psi(v,z)\wp(dz,dv) =  \int_{0}^{t} \int_{\mathbb{R}}\Tilde{G}^\psi(v,z)\tilde{\wp}(dz,dv) + \int_{0}^{t} \int_{\mathbb{R}}\Tilde{G}^\psi(v,z)dz\;dv. \label{genpisxy3}    
\end{equation}\\
Therefore we conclude using \eqref{genpisxy1}-\eqref{genpisxy3} that the infinitesimal generator ${\tilde{\mathcal{A}}}^l$ of the process $(P,I,S,X,Y)$ associated with the constant control taking value $l$ is given by
\begin{align*}
&{\tilde{\mathcal{A}}}^l\psi(p,i,s,x,y)=\diffp{\psi}{s}(p,i,s,x,y)\\
&+\sum_{\nu=\pm,k\in\mathcal{K}}{\lambda}_{\nu}(s)\vartheta_{\nu}(\{k\})(\psi(p,i,s,x+\nu l^\nu k( p(1+\nu\delta)-\nu\epsilon),y-\nu l^\nu k)-\psi(p,i,s,x,y))\\
&+\sum_{j\not\equiv i}h_{ij}(s)(\psi(p(1+\delta\alpha(j)),j,0,x+\alpha(j)l^{\bar{\alpha}(j)}K(p(1+\delta\alpha(j))-\alpha(j)\epsilon),y-\alpha(j)l^{\bar{\alpha}(j)}K)\\
&-\psi(p,i,s,x,y)).
\end{align*}
Since for every $i\in\mathcal{X}$, the set $\{\bar{\alpha}(j)\mid j\not\equiv i\}=\{+,-\}$ therefore we can replace $\nu$ by $\bar{\alpha}(j)$ in the above equation and we obtain
\begin{align}
&{\tilde{\mathcal{A}}}^l\psi(p,i,s,x,y)=\diffp{\psi}{s}(p,i,s,x,y)\nonumber\\
&+\sum_{j\not\equiv i}\bigg({\lambda}_{\bar{\alpha}(j)}(s)\sum_{k\in\mathcal{K}}\vartheta_{\bar{\alpha}(j)}(\{k\})(\psi(p,i,s,x+\alpha(j) l^{\bar{\alpha}(j)} k( p(1+\delta\alpha(j))-\alpha(j)\epsilon),y-\alpha(j) l^{\bar{\alpha}(j)} k)\nonumber\\
&-\psi(p,i,s,x,y))+h_{ij}(s)(\psi(p(1+\delta\alpha(j)),j,0,x+\alpha(j)l^{\bar{\alpha}(j)}K(p(1+\delta\alpha(j))-\alpha(j)\epsilon),\nonumber\\
&y-\alpha(j)l^{\bar{\alpha}(j)}K)-\psi(p,i,s,x,y))\bigg).\label{genpisxy}
\end{align}

\section{Utility function and the HJB equation}
\subsection{The utility function} After the formulation of market-making problem, the next goal is to maximize the risk-sensitive terminal expected portfolio value of the agent with a constant risk-aversion parameter $\eta >0$ taking into account all possible strategies. We aim to maximize the following function $J_{\Hat{l}}$ associated with the control process $\Hat{l}$ where  
\begin{equation}
J_{\Hat{l}}(t,p,i,s,x,y)=E[X_T+P_T Y_T-\eta {Y_T}^2\mid P_t=p,I_t=i,S_t=s,X_t=x,Y_t=y] \label{utility1}
\end{equation}
where $(t,p,i,s,x,y) \in [0,T]\times[0,\infty)\times\mathcal{X}\times[0,\infty)\times\mathbb{R}\times\mathbb{Z}$.\\\\
The market-making optimization problem is to maximize $J_{\Hat{l}}$ over the class $\mathcal{L}$ of all possible strategies and the corresponding utility function is given by
\begin{equation}
J(t,p,i,s,x,y):=\sup_{\Hat{l}\in \mathcal{L}}J_{\Hat{l}}(t,p,i,s,x,y) \label{utility2}
\end{equation}
where $(t,p,i,s,x,y)$ are as in equation \eqref{utility1}.
\begin{lemma}\label{lemma5.1}
For every $(t,p,i,s,x,y)\in [0,T]\times[0,\infty)\times\mathcal{X}\times[0,\infty)\times\mathbb{R}\times\mathbb{Z}$, the utility function $J$ satisfies
\begin{equation}
J(t,p,i,s,x,y)\leq x+yp+Kp\frac{(1+\delta)}{\delta}\left(e^{c(T-t)\delta}-1\right)+Kc(T-t)\epsilon \label{Jbound}   
\end{equation}
where $c=2(K+1)\max{\{c_1,c_2\}}$.
\end{lemma}
\begin{proof}
Referring to Table \ref{fig:portfolio}, we see that the change in portfolio value (when a limit order gets fulfilled) corresponding to every control process $\Hat{l}$ is always less than $K(p\delta+\epsilon)$ where $p$ denotes the mid-price just before the execution of the limit order. In another words for $t\in[0,T]$ and for every $\Hat{l}\in\mathcal{L}$, this means
\[X_{T_{N(t)+1}}+Y_{T_{N(t)+1}}P_{T_{N(t)+1}}\leq X_t+Y_tP_t+K(P_t\delta+\epsilon).\]
Then using the dynamics of process $P$ and the above inequality we obtain
\begin{align}
X_T+Y_TP_T&\leq X_t+Y_tP_t+\sum_{n=1}^{N(T)-N(t)}K(P_t{(1+\delta)}^{n}\delta+\epsilon)\nonumber\\
&= X_t+Y_tP_t+KP_t(1+\delta)\frac{{(1+\delta)}^{N(T)-N(t)}-1}{\delta}+K(N(T)-N(t))\epsilon. \label{XTYTPT1}
\end{align}
We further note that all Poisson points corresponding to the change in portfolio lie in the rectangle $[0,T]\times[0,c]$ and 
\[N(T)-N(t)\leq \wp((t,T]\times[0,c]).\]
Now using \eqref{XTYTPT1} and properties of $\wp$, we get
\begin{align*}
&E[X_T+Y_TP_T\mid (P_t,I_t,S_t,X_t,Y_t)] \\
&\leq X_t+Y_tP_t+KP_t\frac{(1+\delta)}{\delta}\left(E\left[{(1+\delta)}^{\wp((t,T]\times[0,c])}\right]-1\right)+K E[\wp((t,T]\times[0,c])]\epsilon.
\end{align*}
Proceeding similarly as in \eqref{EbNpT1} and using the fact that intensity measure of $\wp$ is the product Lebesgue measure, we obtain
\begin{align}
&E[X_T+Y_TP_T\mid P_t=p,I_t=i,S_t=s,X_t=x,Y_t=y]\nonumber\\
&\leq x+yp+Kp\frac{(1+\delta)}{\delta}\left(e^{c(T-t)\delta}-1\right)+Kc(T-t)\epsilon. \label{XTYTPT2}
\end{align}
Since 
\begin{align*}
&E[X_T+Y_TP_T-\eta {Y_T}^2\mid P_t=p,I_t=i,S_t=s,X_t=x,Y_t=y]\\
&\leq E[X_T+Y_TP_T\mid P_t=p,I_t=i,S_t=s,X_t=x,Y_t=y], 
\end{align*}
and as \eqref{XTYTPT2} is true for every control process $\Hat{l}\in\mathcal{L}$, we have shown that \eqref{Jbound} holds for every $(t,p,i,s,x,y)\in [0,T]\times[0,\infty)\times\mathcal{X}\times[0,\infty)\times\mathbb{R}\times\mathbb{Z}$.
\end{proof}
The above lemma proves that the value function $J$ is well defined and finite.
\subsection{Hamilton-Jacobi-Bellman equation}
For a constant control taking value $l\in L$, we first define the operator
\begin{equation}
{\mathcal{A}}^l:=\diffp{}{t}+\diffp{}{s}+\max_{l\in L}\sum_{j\not\equiv i}\left({\lambda}_{\bar{\alpha}(j)}(s)\sum_{k\in\mathcal{K}}\vartheta_{\bar{\alpha}(j)}(\{k\}){\Delta}_1+h_{ij}(s){\Delta}_2\right) \label{Al}
\end{equation}
where
\begin{align}
{\Delta}_1 \Hat{J}(t,p,i,j,s,x,y,l):&= \Hat{J}(p,i,s,x+\alpha(j) l^{\bar{\alpha}(j)} k( p(1+\delta\alpha(j))-\alpha(j)\epsilon),
y-\alpha(j)l^{\bar{\alpha}(j)} k)\nonumber\\
&-\Hat{J}(t,p,i,s,x,y),\label{Delta1}\\
{\Delta}_2\Hat{J}(t,p,i,j,s,x,y,l):&=\Hat{J}(t,p(1+\delta\alpha(j)),j,0,x+\alpha(j)l^{\bar{\alpha}(j)}K(p(1+\delta\alpha(j))-\alpha(j)\epsilon),\nonumber\\
&y-\alpha(j)l^{\bar{\alpha}(j)}K)-\Hat{J}(t,p,i,s,x,y).\label{Delta2}
\end{align}
for $\Hat{J}$ belonging to the domain of ${\mathcal{A}}^l$.\\\\
Now, the HJB obtained on applying Bellman's principle of dynamic programming for controlled Markov processes (see section III.7, pp 131-133, \cite{fleming}) for the above optimization problem is given by
\begin{equation}
\max_{l \in L}{\mathcal{A}}^l J=0. \label{hjb1}    
\end{equation}
with the terminal condition
\begin{equation}
J(T,p,i,s,x,y)=x+yp-\eta{y^2}\quad \text{for all } (p,i,s,x,y)\in [0,T]\times[0,\infty)\times\mathcal{X}\times[0,\infty)\times\mathbb{R}\times\mathbb{Z}. \label{hjb2}
\end{equation}\\
Moreover in the case when the agent adopts the \emph{holding} strategy at time `$t$', it means she does not place any limit order on either side of the LOB after time $t$. So if $l$ is the control process then we have $l_v=0$ for all $v\in(t,T]$. Now using \eqref{g4}, \eqref{g5},\eqref{wealthfn}, \eqref{inventoryfn}, \eqref{utility1} and \eqref{utility2}, we see that the value function in this case which we denote by $J_0$ is equal to 
\begin{equation}
J_0(t,p,i,s,x,y)=x+y\pi(t,p,i,s)-\eta y^2. \label{J0}   
\end{equation}

\subsection{Solution of HJB equation}~\\\\
We discuss the no risk aversion case when $\eta=0$. The value function in this case is independent of the penalization associated with holding inventory. The agent can hold large inventory in order to maximize her return without worrying about any kind of risks involved with holding large inventory such as holding costs incurred, possibility of degradation with time etc.
\begin{theorem}\label{eta0}
For $\eta=0$, the utility function defined in \eqref{utility2} is given by
\begin{equation}
J^{(0)}(t,p,i,s,x,y)=x+y\pi(t,p,i,s)+u(t,p,i,s) \label{J(0)}    
\end{equation}
where
\begin{equation}
u(t,p,i,s)=\sum_{j\not\equiv i}E\left[\int_t^T \max{\{m_j(v,P_{v^-},I_{v^-},S_{v^-}),0\}dv}\bigg\vert P_t=p,I_t=i,S_t=s\right]\label{u}
\end{equation}
and
\begin{align}
m_j(t,p,i,s)&={\lambda}_{\bar{\alpha}(j)}(s)\sum_{k\in\mathcal{K}}k{\vartheta}_{\bar{\alpha}(j)}(\{k\})(\alpha(j) p-\alpha(j)\pi(t,p,i,s)+(\delta-\epsilon))\nonumber\\
&+h_{ij}(s)K(\alpha(j)p-\alpha(j)\pi(t,p(1+\delta\alpha(j)),j,0)+(\delta-\epsilon)).\label{mj}      
\end{align}
Furthermore, the optimal control denoted by ${\Hat{l}}^{(0)}$ in this case is given by
\begin{equation}
{\Hat{l}}^{(0)}(t,p,i,s)=(1_{m_{j_1}(t,p,i,s)>0},1_{m_{j_2}(t,p,i,s)>0}) \label{Hatl0}
\end{equation}
where $\{j_1,j_2\}=\{j\in\mathcal{X}\mid j\not\equiv i\}$.
\end{theorem}
\begin{proof}
Suppose the solution to the HJB equation \eqref{hjb1} with the terminal condition \eqref{hjb2} exists. We take a solution say $J^{(0)}$, of \eqref{hjb1}-\eqref{hjb2} of the form $J^{(0)}(t,p,i,s,x,y)=x+y\pi(t,p,i,s)+u(t,p,i,s)$ where $u$ is a function that needs to be determined. Now using the similar notations as in \eqref{Delta1} and \eqref{Delta2}, we compute
\begin{align}
&{\Delta}_1J^{(0)}(t,p,i,j,s,x,y,l)=l^{\bar{\alpha}(j)}(\alpha(j)k p-\alpha(j)k\pi(t,p,i,s)+k(\delta-\epsilon)),\label{Delta1J0}\\
&{\Delta}_2J^{(0)}(t,p,i,j,s,x,y,l)=l^{\bar{\alpha}(j)}(\alpha(j)Kp-\alpha(j)K\pi(t,p(1+\delta\alpha(j)),j,0)+K(\delta-\epsilon))\nonumber\\
&+y(\pi(t,p(1+\delta\alpha(j)),j,0)-\pi(t,p,i,s))+u(t,p(1+\delta\alpha(j)),j,0)-u(t,p,i,s)).\label{Delta2J0}
\end{align}
Also since,
\begin{align*}
\diffp{J^{(0)}}{t}+\diffp{J^{(0)}}{s}=\pi\left(\diffp{\pi}{t}+\diffp{\pi}{s}\right)+\diffp{u}{t}+\diffp{u}{s}
\end{align*}
Using corollary \ref{corollary3.1}, the expressions of ${\Delta}_1J^{(0)}$ and ${\Delta}_2J^{(0)}$ we obtain that $J^{(0)}$ satisfies \eqref{hjb1}-\eqref{hjb2} if and only if
\begin{align}
&\diffp{u}{t}+\diffp{u}{s}+\sum_{j\not\equiv i}h_{ij}(s)(u(t,p(1+\delta\alpha(j)),j,0)-u(t,p,i,s)))+\max_{l\in L}\sum_{j\not\equiv i}l^{\bar{\alpha}(j)}m_j(t,p,i,s)=0,\label{hjbeta01}\\
&u(T,p,i,s)=0\quad \text{for every } (p,i,s) \in[0,\infty)\times\mathcal{X}\times[0,\infty).\label{hjbeta02}
\end{align}
where
\begin{align*}
m_j(t,p,i,s):&={\lambda}_{\bar{\alpha}(j)}(s)\sum_{k\in\mathcal{K}}k{\vartheta}_{\bar{\alpha}(j)}(\{k\})(\alpha(j) p-\alpha(j)\pi(t,p,i,s)+(\delta-\epsilon))\\
&+h_{ij}(s)K(\alpha(j)p-\alpha(j)\pi(t,p(1+\delta\alpha(j)),j,0)+(\delta-\epsilon)).
\end{align*}
Since $l^+,l^- \in \{0,1\}$, we note that \eqref{hjbeta01} can be written as
\begin{equation}
\diffp{u}{t}+\diffp{u}{s}+\sum_{j\not\equiv i}h_{ij}(s)(u(t,p(1+\delta\alpha(j)),j,0)-u(t,p,i,s)))+\sum_{j\not\equiv i}\max{\{m_j(t,p,i,s),0\}}=0. \label{hjbeta03}    
\end{equation}
Now using the assumptions (A2) and (A7) and the expression of $m_j$, we have
\begin{align*}
\frac{|m_j(t,p,i,s)|}{1+p}&\leq (c_2{\vartheta}^1_{\bar{\alpha}(j)}+c_1K)\frac{p+|\delta-\epsilon|}{1+p}+c_2{\vartheta}^1_{\bar{\alpha}(j)}\frac{|\pi(t,p,i,s)|}{1+p}+c_1K\frac{|\pi(t,p(1+\delta\alpha(j)),j,0)|}{1+p}\\
&\leq (c_2{\vartheta}^1_{\bar{\alpha}(j)}+c_1K)\frac{p+|\delta-\epsilon|}{1+p}+c_2{\vartheta}^1_{\bar{\alpha}(j)}\frac{|\pi(t,p,i,s)|}{1+p}\\
&+c_1K\frac{p\delta}{1+p}\frac{|\pi(t,p(1+\delta\alpha(j)),j,0)|}{1+p(1+\delta\alpha(j))}+c_1K\frac{|\pi(t,p(1+\delta\alpha(j)),j,0)|}{1+p(1+\delta\alpha(j))}.
\end{align*}
As $\pi$ belongs to $V$, using the above inequality we see that $\sup_{(t,p,i,s)\in\mathcal{D}}\frac{|m_j(t,p,i,s)|}{1+p}<\infty.$ Moreover using assumptions (A1), (A8) and the continuity of $\pi$ it follows from \eqref{mj} that $m_j$ is a continuous function for each $j\not\equiv i$. Therefore we have shown that $m_j$ belongs to $V$ for each $j\not\equiv i$. Hence, the function $\sum_{j\not\equiv i}\max{\{m_j(t,p,i,s),0\}}$ belongs to $V$.\\\\
Thus, using theorem \ref{thm2} we conclude that the partial differential equation \eqref{hjbeta03} with the terminal condition \eqref{hjbeta02} has a unique solution in the class $V\cap cl({\mathcal{D}}_{t,s})$ which is given by (using \eqref{storep})
\begin{equation*}
u(t,p,i,s)=\sum_{j\not\equiv i}E\left[\int_t^T \max{\{m_j(v,P_{v^-},I_{v^-},S_{v^-}),0\}dv}\bigg\vert P_t=p,I_t=i,S_t=s\right].
\end{equation*}
This implies that the HJB equation \eqref{hjb1} with the terminal condition \eqref{hjb2} in the case when $\eta=0$ has a unique solution (say $J^{(0)}$) in the class $V\cap cl({\mathcal{D}}_{t,s})$ where
\begin{align*}
&J^{(0)}(t,p,i,s,x,y)\\
&=x+y\pi(t,p,i,s)+\sum_{j\not\equiv i}E\left[\int_t^T \max{\{m_j(v,P_{v^-},I_{v^-},S_{v^-}),0\}dv}\bigg\vert P_t=p,I_t=i,S_t=s\right].  
\end{align*}
Further since $J^{(0)}$ given by the above equation belongs to $V\cap cl({\mathcal{D}}_{t,s})$, it clearly belongs to the domain of the operator ${\mathcal{A}}^l$ for every constant control $l$. Hence $J^{(0)}$ is a classical solution (see section III.8, pp 134, \cite{fleming}) of the equation \eqref{hjb1} with the terminal condition \eqref{hjb2}. Moreover, since the maximum over the set $L$ in \eqref{hjb1} is attained when $m_j(t,p,i,s)>0$ for each $j\not\equiv i$, i.e., the optimal control ${\Hat{l}}_0$ is given by
\begin{equation*}
{\Hat{l}}^{(0)}(t,p,i,s)=(1_{m_{j_1}(t,p,i,s)>0},1_{m_{j_2}(t,p,i,s)>0})
\end{equation*}
where $\{j_1,j_2\}=\{j\in\mathcal{X}\mid j\not\equiv i\}$.\\
Therefore, it follows from Theorem 8.1, pp 135, \cite{fleming}, that $J^{(0)}$ is an optimal solution to \eqref{hjb1}-\eqref{hjb2} with the optimal control given by \eqref{Hatl0} and thus it is indeed the utility function in the no risk aversion case.
\end{proof}
\section{Conclusion}
In this paper, we modeled the stock price dynamics and the portfolio of a market maker who posts limit orders continuously at the best bid and the best ask prices, using stochastic differential equations obtained through a Poisson random measure. In \cite{pham2}, the expected value of the stock price at the terminal time, the value function and other quantities were shown to be satisfying the corresponding partial differential equations in the viscosity sense which is a weaker notion whereas we obtained them as the classical solution of the corresponding pde. A possible direction to further explore is to study and solve the HJB equation obtained for various other standard utility functions such as the isoelastic utility function, exponential utility, logarithmic utility etc.

\appendix
\centering\section{}
\begin{lemma} \label{lemmaA1}
Under the assumption (A6), the probability distribution of $(P_t,I_t,S_t)$ has full support for every $t \in [0,T]$.
\end{lemma}
\begin{proof}
Consider $(t,p,i,s)\in\mathcal{D}$ and $\varepsilon>0$. Then
\begin{align*}
&\mathbb{P}(P_t\in(p-\varepsilon,p+\varepsilon),I_t=i,S_t\in(s-\varepsilon,s+\varepsilon),T_1>t)\\ =&\mathbb{P}(P_0\in(p-\varepsilon,p+\varepsilon),I_0=i,S_0\in(s-t-\varepsilon,s-t+\varepsilon))(1-F(t)).  
\end{align*}
Now when $s\geq t$, using assumption (A6) and the fact that $F$ never attains value 1 we conclude
\begin{equation}
\mathbb{P}(P_t\in(p-\varepsilon,p+\varepsilon),I_t=i,S_t\in(s-\varepsilon,s+\varepsilon),T_1>t)>0. \label{PtItStT1t1}
\end{equation} 
Further for $t>0$, we calculate
\begin{align}
\mathbb{P}(T_1\leq t,T_2-T_1>t)&=\mathbb{P}(T_2-T_1>t)\mathbb{P}(T_1 \leq t)\nonumber\\
&=(1-F(t))F(t)>0.\label{T1T2T1t}
\end{align}
Now when $s<t$, using \eqref{T1T2T1t} we can write
\begin{align}
&\mathbb{P}(P_t\in(p-\varepsilon,p+\varepsilon),I_t=i,S_t\in(s-\varepsilon,s+\varepsilon),T_1\leq t,T_2-T_1>t)\nonumber\\
=&\mathbb{P}(P_t\in(p-\varepsilon,p+\varepsilon),I_t=i,S_t\in(s-\varepsilon,s+\varepsilon)\mid T_1\leq t,T_2-T_1>t)(1-F(t))F(t)\nonumber\\
=&\mathbb{P}(P_0(1+\delta\alpha(I_{T_1}))\in(p-\varepsilon,p+\varepsilon),I_{T_1}=i,T_1\in (t-s-\varepsilon,min\{t,t-s+\varepsilon\}))(1-F(t))F(t)\nonumber\\
=&\mathbb{P}(P_0(1+\delta\alpha(T_1))\in(p-\varepsilon,p+\varepsilon),I_{T_1}=i)F(min\{t,t-s+\varepsilon\}-F(t-s-\varepsilon))(1-F(t))F(t)\nonumber\\
=&\mathbb{P}\bigg(P_0\in\bigg(\frac{p-\varepsilon}{1+\delta\alpha(i)},\frac{p+\varepsilon}{1+\delta\alpha(i)}\bigg)\bigg)F(min\{t,t-s+\varepsilon\}-F(t-s-\varepsilon))(1-F(t))F(t)\cdot\nonumber\\
&\sum_{j\not\equiv i}{\Hat{p}}_{ij}\mathbb{P}(I_0=j)>0 \label{PtItStT1t2}
\end{align}
using the assumption (A6), the expression of ${\Hat{p}}_{ij}$ and the fact that $F$ is a strictly increasing function that never attains the value $1$. Now since
\begin{align*}
&\mathbb{P}(P_t\in(p-\varepsilon,p+\varepsilon),I_t=i,S_t\in(s-\varepsilon,s+\varepsilon))\\
\geq&\mathbb{P}(P_t\in(p-\varepsilon,p+\varepsilon),I_t=i,S_t\in(s-\varepsilon,s+\varepsilon),T_1>t)\\
+&\mathbb{P}(P_t\in(p-\varepsilon,p+\varepsilon),I_t=i,S_t\in(s-\varepsilon,s+\varepsilon),T_1\leq t,T_2-T_1>t),
\end{align*}
using \eqref{PtItStT1t1} and \eqref{PtItStT1t2} in the case when $s\geq t$ and $s<t$ respectively, we conclude that for $(t,p,i,s)\in\mathcal{D},\varepsilon>0$, 
\[\mathbb{P}(P_t\in(p-\varepsilon,p+\varepsilon),I_t=i,S_t\in(s-\varepsilon,s+\varepsilon))>0.\]
Hence we have proved that the distribution of $(P_t,I_t,S_t)$ has full support for every $t\in[0,T]$.
\end{proof}

\nocite{*}
  \bibliography{Ref.bib}

\begin{thebibliography}{1}

\bibitem{nber}
{\em {NBER} working group descriptions}.
\newblock
  \url{https://web.archive.org/web/20080722025938/http://www.nber.org/workinggroups/groups_desc.html},
  Jul 2008.

\bibitem{cinlar2011probability}
{\sc E.~{\c{C}}{\i}nlar}, {\em Probability and Stochastics}, Graduate Texts in
  Mathematics, Springer New York, NY, 2011.

\bibitem{fleming}
{\sc W.~Fleming and H.~Soner}, {\em Controlled Markov Processes and Viscosity
  Solutions}, Stochastic Modelling and Applied Probability, Springer New York,
  2~ed., 2006.

\bibitem{pham2}
{\sc P.~Fodra and H.~Pham}, {\em High frequency trading and asymptotics for
  small risk aversion in a markov renewal model}, SIAM Journal on Financial
  Mathematics, 6 (2015), pp.~656--684.

\bibitem{pham1}
{\sc P.~Fodra and H.~Pham}, {\em Semi-markov model for market microstructure},
  Applied Mathematical Finance, 22 (2015), pp.~261--295.

\bibitem{Ghosh2009}
{\sc M.~K. Ghosh, A.~Goswami, and S.~K. Kumar}, {\em Portfolio optimization in
  a semi-markov modulated market}, Applied Mathematics and Optimization, 60
  (2009), pp.~275--296.

\bibitem{mkghoshsaha}
{\sc M.~K. Ghosh and S.~Saha}, {\em Stochastic processes with age-dependent
  transition rates}, Stochastic Analysis and Applications, 29 (2011),
  pp.~511--522.

\bibitem{ajp}
{\sc A.~Goswami, J.~Patel, and P.~Shevgaonkar}, {\em A system of non-local
  parabolic pde and application to option pricing}, Stochastic Analysis and
  Applications, 34 (2016), pp.~893--905.

\bibitem{pazy}
{\sc A.~Pazy}, {\em Semigroups of Linear Operators and Applications to Partial
  Differential Equations}, Applied mathematical sciences, Springer New York,NY,
  1~ed., 1983.

\end{thebibliography}
  \bibliographystyle{siam}
\end{document}